\documentclass[twocolumn,pra, notitlepage, preprintnumbers,nopacs]{revtex4-1}
\newcommand{\pagenumbaa}{1}
\usepackage{graphicx}
\usepackage{amssymb}
\usepackage{amstext}
\usepackage{epsfig}
\usepackage[colorlinks,pdftex]{hyperref} 

\usepackage{latexsym}
\usepackage{amsmath}
\usepackage{bm} 
\usepackage{bbm}

\usepackage{amsthm}
\usepackage{amsmath}
\usepackage{amsfonts}
\usepackage{amssymb,amstext}
\usepackage{appendix}
\usepackage{enumitem}

\theoremstyle{plain}
\newtheorem{theorem}{Theorem}
\newtheorem{lemma}[theorem]{Lemma}
\newtheorem{corollary}[theorem]{Corollary}

\makeatletter
\newtheorem*{rep@theorem}{\rep@title}
\newcommand{\newreptheorem}[2]{%
\newenvironment{rep#1}[1]{%
 \def\rep@title{#2 \ref{##1} (restatement)}%
 \begin{rep@theorem}}%
 {\end{rep@theorem}}}
\makeatother

\newreptheorem{theorem}{Theorem}
\newreptheorem{lemma}{Lemma}

\theoremstyle{definition}

\usepackage{epsfig}
\usepackage{algpseudocode}
\usepackage{amscd}
\usepackage{algorithm}
\usepackage{tikz}

\newcommand{\beq}{\begin{equation}}
\newcommand{\eeq}{\end{equation}}

\newcommand{\beqa}{\begin{eqnarray}}
\newcommand{\eeqa}{\end{eqnarray}}

\newcommand{\bal}{\begin{aligned}}
\newcommand{\eal}{\end{aligned}}

\newcommand{\bsp}{\begin{equation}\begin{split}}
\newcommand{\esp}{\end{split}\end{equation}}

\newcommand{\bit}{\begin{itemize}}
\newcommand{\eit}{\end{itemize}}

\newcommand{\ben}{\begin{enumerate}}
\newcommand{\een}{\end{enumerate}}
\newcommand{\nn}{\nonumber}

\newcommand*{\ket}[1]{| #1 \rangle}
\newcommand*{\bra}[1]{\langle #1 |}

\newcommand{\ketbra}[1]{| #1 \rangle \langle #1 |}
\newcommand{\id}{I }
\newcommand{\tr}{\mathrm{tr}}
\newcommand{\rank}{\mathrm{rank}}
\newcommand{\braket}[2]{\langle #1 | #2 \rangle}

\newcommand{\data}{\mathcal{D}}

\newcommand{\comp}{\mathcal{C}}
\newcommand{\prob}[1]{\mathbb{P}\bigl[ #1 \bigr]}

\newcommand{\nc}{\newcommand}
\nc{\rnc}{\renewcommand}

\newcommand{\proj}[1]{\left|#1\right\rangle\left\langle #1\right|}

\def\be#1\ee{\begin{equation}#1\end{equation}}
\def\bea#1\eea{\begin{eqnarray}#1\end{eqnarray}}
\def\beas#1\eeas{\begin{eqnarray*}#1\end{eqnarray*}}
\def\ba#1\ea{\begin{align}#1\end{align}}
\def\bas#1\eas{\begin{align*}#1\end{align*}}
\def\bpm#1\epm{\begin{pmatrix}#1\end{pmatrix}}

\def\nn{\nonumber}
\def\eq#1{(\ref{eq:#1})}

\def\L{\left} 
\def\R{\right}

\def\eps{\varepsilon}

\def\cC{\mathcal{C}}
\def\cD{\mathcal{D}}
\def\cE{\mathcal{E}}

\def\cJ{\mathcal{J}}

\def\cO{{\cal O}}

\def\bbC{\mathbb{C}}
\DeclareMathOperator*{\E}{\mathbb{E}}
\DeclareMathOperator*{\bbE}{\mathbb{E}}

\def\benum{\begin{enumerate}}
\def\eenum{\end{enumerate}}
\def\bit{\begin{itemize}}
\def\eit{\end{itemize}}

\newcommand{\secref}[1]{Section~\ref{sec:#1}}

\newcommand{\lemref}[1]{Lemma~\ref{lem:#1}}
\newcommand{\thmref}[1]{Theorem~\ref{thm:#1}}

\newcommand{\corref}[1]{Corollary~\ref{cor:#1}}

\nc{\todo}[1]{\textcolor{red}{todo: #1}}

\def\begsub#1#2\endsub{\begin{subequations}\label{eq:#1}\begin{align}#2\end{align}\end{subequations}}
\nc\qand{\qquad\text{and}\qquad}
\nc\mnb[1]{\medskip\noindent{\bf #1}}

\begin{document}

\hspace{4.6cm} \vbox{
\hbox{MIT--CTP 4619}
}

\vspace*{2ex}

\title{\boldmath Extracting short distance information from $b\to s\,\ell^+\ell^-$ effectively}


\title{Compressibility of positive semidefinite factorizations and quantum models}


\author{Cyril J. Stark}
\author{Aram W. Harrow}

\affiliation
{Center for Theoretical Physics, Massachusetts Institute of Technology, 77 Massachusetts Avenue, Cambridge MA 02139-4307, USA}

\date{\today}


\begin{abstract}

We investigate compressibility of the dimension of positive semidefinite matrices while approximately preserving their pairwise inner products. This can either be regarded as compression of positive semidefinite factorizations of nonnegative matrices or (if the matrices are subject to additional normalization constraints) as compression of quantum models. We derive both lower and upper bounds on compressibility. Applications are broad and range from the statistical analysis of experimental data to bounding the one-way quantum communication complexity of Boolean functions.
\preprint{MIT-CTP/4619}
\end{abstract}


\maketitle

\setcounter{page}{\pagenumbaa}
\thispagestyle{plain}

\section{Introduction}

The following situation is ubiquitous in quantum information: a $d$-dimensional state $\rho_x$ is prepared and a measurement $E_y$ is performed with POVM
elements $E_{y,z}$, with $x \in [X], y\in [Y], z\in [Z]$.  The
resulting conditional probability distribution of these outcomes
$p(z|x,y) = \tr (\rho_x E_{y,z})$ can be expressed as a matrix $\cD$
with $\cD_{x; y,z} = p(z|x,y)$.  The forward problem of computing
$\cD$ given $\{\rho_x\}, \{E_y\}$ is straightforward, but often we
need to solve the {\em inverse} problem of finding states and
measurements compatible with a given $\cD$.    This problem of finding
{\em quantum models} is in general
underdetermined, but we can constrain the problem by minimizing the dimension
$d$ of the model.  Similarly, if given an accuracy parameter $\eps$, we could ask
for the minimum $d$ for which a $d$-dimensional quantum model can approximate each entry of $\cD$ to
additive error $\eps$.   

Two important examples where this problem occurs are:
\bit
\item {\em Inferring quantum models.}  Suppose we perform an
  experiment where $x$ and $y$ are classical inputs (e.g. $x$ selects
  the state and $y$ the measurement) and $z$ is
  classical observable.  We will sample from the distributions
  $p(z|x,y)$ and thereby learn an approximation of $\cD$.  However, we
  do not have direct access to the underlying quantum states or measurements.  In
  reconstructing these states and measurements, it is natural (e.g., to prevent overfitting) to posit the simplest
  model consistent with the data, which in the absence of other
  information would mean the quantum model of lowest dimension.
\item {\em One-way quantum communication complexity.}  Suppose that Alice
  and Bob would like to jointly compute a function  $f :[X] \times
  [Y]\mapsto [Z]$ when Alice is given input $x$ and Bob is given
  input $y$.  The most general protocol consists of Alice sending
  $\rho_x$ to Bob who performs POVM $E_y$ and outputs the outcome $z$
  that he obtains.  The log of the minimum dimension achieving $p(z|x,y) =
  \delta_{z, f(x,y)}$ (resp. $p(z|x,y)  \approx   \delta_{z, f(x,y)}$)
 is the exact (resp. approximate) one-way quantum
  communication complexity of $f$.
\eit

Quantum models for conditional probability distributions are an
example of the more general idea of a psd (positive semidefinite)
factorization.  If $M$ is a matrix with nonnegative entries, then a
psd factorization of $M$ is a collection of $d$-dimensional psd
matrices $\{A_x\}, \{B_y\}$ such that $M_{xy} = \tr (A_xB_y)$.  The
psd-rank of $M$ is the minimum $d$ for which this is possible;
approximate versions can also be defined~\cite{fiorini2012linear,gouveia2013approx}.  See
\cite{fawzi2014positive} for a recent review of psd-rank.  One
application of psd-rank is to the case when $M$ is the slack matrix of
a polytope, in which case the psd-rank is the smallest possible
dimension of a semidefinite program representing the polytope; here
too approximate versions of this relation are
known~\cite{gouveia2013approx}.  Our problem of finding
low-dimensional quantum models is a special case of the general
psd-rank problem which differs in our requirements for normalization:
$\rho_x$ should not only be psd but also trace 1, and $E_{yz}$ should
not only be psd but should satisfy $\sum_z E_{yz} = I$ for all $y$.

In this paper we study the question of when a quantum model can be
compressed to a smaller dimension.  We give conditions under which
compression is, or is not, possible.   \\

\emph{Main results.} We present one theorem demonstrating {\em
  incompressibility} (Theorem~\ref{Thm:incompressibility}) and three theorems describing {\em compressibility} (Theorems~\ref{thm:comp.of.psd.factorizations},~\ref{thm:comp.of.q.models} and~\ref{thm:comp.of.q.models.with.spectral.tails}).  A
common theme will be that measurement operators with high trace are a barrier to
compressibility.

\begin{theorem}\label{Thm:incompressibility}
	Let $\data \in \mathbb{R}^{X \times YZ}$, let $\bigl( E_{yz}
        \bigr)_{z=1}^Z$ be a $d$-dimensional POVMs and let $\rho_x$, $x
        \in [X]$, be $d$-dimensional quantum states. Assume that
        $\data_{x,yz} = \tr\bigl( \rho_x E_{yz} \bigr)$ for all $x \in [X]$, $y \in [Y]$
        and $z \in [Z]$. Fix $y$. For all $z \in [Z]$, set $c^*_z = \max\{
        \data_{x;yz} \}_{x=1}^X$.  Then, $d \geq \sum_{z=1}^Z c^*_z$.
        Moreover, if $\data' \in \mathbb{R}^{X \times YZ}$ is a
        set of observations coming from $d'$-dimensional states that satisfies
        $\| \data - \data' \|_{\infty} \leq \eps$, then we have that
       $d' \geq -Z \eps + \sum_{z=1}^Z c^*_z$.
\end{theorem}

Theorem~\ref{Thm:incompressibility} is proven in section~\ref{Sect:incompressibility}. Theorem~\ref{Thm:incompressibility} was obtained independently in~\cite{lee2014some} (Theorem~24).

Note that the lower bound from Theorem~\ref{Thm:incompressibility}
cannot exceed $Z$, i.e., the number of measurement
outcomes. Consequently, Theorem~\ref{Thm:incompressibility} leaves
open the possibility to compress quantum models of dimension $D > Z$
into models of dimension $Z$. This is indeed possible if for each measurement $E_y$ all but one POVM element have trace norm constant in $D$. To approach this conclusion we first consider the compression of psd factorizations of $\data$.

\nc{\ThmCompPsdFact}[1]{
	Let $M_1,\ldots,M_J$ be psd matrices on $\mathbb{C}^D$ let $\varepsilon \in (0,1/2]$ and fix $d \in \mathbb{N}$ such that
\be
		d > \frac{16}{\varepsilon^2} \ln\bigl( 2JD \bigr).
\label{eq:d-for-any-M-#1}\ee
	Then there exist psd matrices $M_j'$ on $\mathbb{C}^d$ such that for all $i,j \in [J]$
	\beq\begin{split}\label{eq:matrix-compression-error-#1}
		& \tr(M_{i}M_{j})  - \tr(M_{i}) \tr(M_{j}) 192 \varepsilon \\
		&\ \ \ \ \ \ \leq	\tr\bigl( M'_{i}   M'_{j}  \bigr)  
		\leq	   \tr(M_{i}M_{j})  + \tr(M_{i}) \tr(M_{j}) 192 \varepsilon .
	\end{split}\eeq
	The matrices $M'_j$ can be computed in randomized polynomial time in $J$ and $D$.
}

\begin{theorem}\label{thm:comp.of.psd.factorizations}
\ThmCompPsdFact{first}
\end{theorem}

Theorem~\ref{thm:comp.of.psd.factorizations} is proven in section~\ref{Sect:Compression.of.positive.semidefinite.factorization}. It can be used to compress psd factorizations of $\data$ into approximate psd factorizations of $\data$. In particular, it can be used to compress quantum models for $\data$ into approximate psd factorizations of $\data$. However, we have no guarantee that the compression satisfies the normalization conditions which played the key role in the derivation of Theorem~\ref{Thm:incompressibility}. This problem is addressed in Theorem~\ref{thm:comp.of.q.models}.

\begin{theorem}\label{thm:comp.of.q.models}
	For each $x \in [X]$, let $\rho_x$ be a $D$-dimensional quantum state and for all $y \in [Y]$ let $\bigl( E_{yz} \bigr)_{z=1}^Z$ be a $D$-dimensional POVM. Set $J = X + YZ$. Let $\varepsilon \in (0,1/2]$ and fix $d \in \mathbb{N}$ such that for all $y$ 
	\beq\bal \label{eq:d.big.enough}
		d &	> \frac{32}{\varepsilon^2}  \ln(4JD) \\
		d &	> \frac{32}{\varepsilon^2} \, \rank\Bigl( \sum_{z=1}^{Z-1} E_{yz} \Bigr)
\eal\eeq
	Then, there exist $d$-dimensional quantum states $\rho'_x$ and $d$-dimensional POVMs $\bigl( E'_{yz} \bigr)_{z=1}^Z$ such that
	\begin{itemize}
          
          \item		for all $x \in [X], y \in [Y], z \in [Z-1]$,
          $$\bigl| \tr\bigl( \rho'_{x} E'_{yz} \bigr) - \tr\bigl( \rho_{x} E_{yz} \bigr) \bigr| \leq 200 \varepsilon \; \tr(E_{yz}).$$ 
          \item	For all $x \in [X], y \in [Y]$ and $z = Z$, $$\bigl| \tr\bigl( \rho'_{x} E'_{yz} \bigr) - \tr\bigl( \rho_{x} E_{yz} \bigr) \bigr| \leq 200 \varepsilon \; \tr(I - E_{yZ}).$$
          
          
	\end{itemize}
	The compressed quantum model $\bigl( \rho'_x \bigr)_{x}$, $\bigl( E'_{yz} \bigr)_{yz}$ can be computed in randomized polynomial time in $X,Y,Z$ and $D$.
\end{theorem}

Similar approximation guarantees hold for overlaps between pairs of states and pairs of measurements (see section~\ref{Sect:tail.bounds.on.Q.to.Q.compression}).

In the above theorem the outcome $z=Z$ is special, and our bounds are adapted to the case when $E_{yZ}$ has much larger rank than the other POVM elements. The practical relevance of measurements of that kind is discussed in section~\ref{Sect:Implications} in the context of different applications of Theorem~\ref{thm:comp.of.q.models}.

Theorem~\ref{thm:comp.of.q.models} is proven in
section~\ref{sec:description.of.Q.to.Q.compression.scheme} to
section~\ref{Sect:q.models:union.bound}. Note that the main difference
between the conditions in Theorem~\ref{thm:comp.of.q.models} and the
conditions in Theorem~\ref{thm:comp.of.psd.factorizations} is the new
rank-based constraint
in~\eqref{eq:d.big.enough}. It is the
direct consequence of the normalization conditions for quantum
models. As an example consider one particular measurement $E_y$ of the
form $E_{yz} = \ketbra{z}$ for $z \in [Z-1]$ and $E_{yZ} = I -
\sum_{z=1}^{Z-1} E_{yz}$. Assume that the experimental states $\rho_x$
are approximately equal to $\ketbra{x}$. Then, by
Theorem~\ref{Thm:incompressibility}, $d \geq Z$. On the other hand, by
Theorem~\ref{thm:comp.of.q.models}, there exists an approximate
quantum model with 
\beq d = \frac{32}{\varepsilon^2} \ln(JD) +
\frac{32}{\varepsilon^2} Z
\eeq
i.e., the lower bound is almost achieved.\\

Condition~\eqref{eq:d.big.enough} is phrased in terms of the rank of POVM elements. This is the main caveat of Theorem~\ref{thm:comp.of.q.models} because in experimental systems we expect POVM elements to be full rank with light spectral tails. The following Theorem~\ref{thm:comp.of.q.models.with.spectral.tails} is another version of Theorem~\ref{thm:comp.of.q.models} which tolerates POVM elements with high rank if their spectrum is decaying exponentially. See section~\ref{Sect:Verification.that.compression.of.POVM.is.psd.and.normalized:incorporating.spectral.tails} for its proof.

\nc{\ThmDecayingTails}[1]{	Let $\rho_x$, $E_{yz}$ and $J$ be as in Theorem~\ref{thm:comp.of.q.models}. Denote by $\bigl( \varepsilon^{(y)}_n \bigr)_n$ the spectrum (ordered descendingly) of $\sum_{z=1}^{Z-1} E_{yz}$. 
To describe potentially thin spectral tails we assume that for some $j^*,b>0$ and for all $j>0$,
\beq\label{eq:tail.bound.on.m.j-#1}
	\varepsilon_{j^*+j} \leq e^{-bj}.
\eeq
Let $\varepsilon \in (0,1/2]$ and fix $d \in \mathbb{N}$ such that  
\ba
	d & \geq \frac{128}{\eps^2}\ln(4JD) \label{eq:JL-tails-#1}\\
		d & \geq \frac{128}{\eps^2} \L(j^*+\frac{1}{b}\ln\frac{8}{\eps}\R) \label{eq:d-beats-tails-#1}
\ea
	Then, there exist $d$-dimensional quantum states $\rho'_x$ and $d$-dimensional POVMs $\bigl( E'_{yz} \bigr)_{z=1}^Z$ with the same approximation promises as Theorem~\ref{thm:comp.of.q.models}. Again, the compressed quantum model $\bigl( \rho'_x \bigr)_{x}$, $\bigl( E'_{yz} \bigr)_{yz}$ can be computed in randomized polynomial time in $X,Y,Z$ and $D$.
}
\begin{theorem}\label{thm:comp.of.q.models.with.spectral.tails}
\ThmDecayingTails{first}
\end{theorem}

\emph{Implications.} Bounding the approximate psd-rank of identity matrices is a problem addressed in~\cite{lee2014some} (section~6.2) and a lower bound was derived. Theorem~\ref{thm:comp.of.psd.factorizations} on the other hand provides upper bounds for general matrices.

Theorem~\ref{thm:comp.of.q.models}
and~\ref{thm:comp.of.q.models.with.spectral.tails} have the potential
to alter our perspective in quantum tomography in that it reveals an
equivalence between high-dimensional models and low-dimensional models
if the measurements are low rank and have few outcomes. This is discussed in section~\ref{Sect:implications.for.quantum.tomo}. 

Another lesson to be learned here is that whether we interpret
measurement outcomes as part of a single measurement or whether we
summarize small groups of measurement outcomes to individual
measurements hugely affects the compressibility of the model. To
illustrate this point, let $\{ \ket{j} \}_{j=1}^D$ denote an
orthonormal basis in the Hilbert space $\mathbb{C}^D$, let $X=D$ with
$\rho_x = \ketbra{x}$, and let $Y=1$, $Z=D$ with $E_{1z} =
\ketbra{z}$.
It follows that $\bigl( \tr(\rho_{x}E_{yz}) \bigr)_{xyz}$
equals the identity matrix of size $D$. Then,
by Theorem~\ref{Thm:incompressibility}, the dimension of each model describing
$\bigl( \tr(\rho_{x}E_{yz}) \bigr)_{xyz}$ is lower bounded by $D$ and
cannot be compressed.  Even allowing an error of $\leq\delta$ in each
entry cannot change the lower bound by more than $D\delta$ which
implies that the dimension must be $\geq (1-\delta)D$.
Indeed this demonstrates that the dependence of $d$ on the rank of the sum of the POVM operators in 
\thmref{comp.of.q.models} cannot in general be removed.  
On the other hand, if instead we performed $D$
individual measurements $\bigl( \ketbra{z}, \id - \ketbra{z} \bigr)$
to test for each of the $D$  possible inputs then the model can be compressed
exponentially. Consequently, it becomes impossible to resolve any high
complexity of the considered physical system. 

The two models described in the last paragraph correspond to the
communication tasks of sending and identifying, respectively, a
$\log(D)$-bit classical message. For these models Winter's 2004 paper~\cite{winter2004quantum}
observed the same exponential compression that we describe.  Our work
can be thought of as a generalization of \cite{winter2004quantum} in
that it replaces the simulation of orthonormal states and binary measurements $\{\ketbra{z}, \id
-\ketbra{z}\}$ with the simulation of general states and measurements. It is qualitatively similar in that it works best when compressing
low-rank measurements.

Further applications of Theorem~\ref{thm:comp.of.q.models} involve a Corollary of Theorem~\ref{thm:comp.of.q.models} (see \secref{quant.1.way.complexity}) that allows for shrinkage of upper bounds on the one-way quantum communication complexity $Q^1(f)$ if the upper bound is based on high-rank measurements. Moreover, we discuss fundamental
limitations on robust dimension witnessing~\cite{acin2007device,harrigan2007representing,wehner2008lower,gallego2010device,stark2014self,dall2012robustness}
in \secref{limits.for.dimension.witnessing}.

\section{Incompressibility}\label{Sect:incompressibility}

Consider a $D$-dimensional quantum model for $\data$. Due to normalization of measurements, $\sum_z E_{yz} = I$. On a $D$-dimensional Hilbert space the identity matrix has trace-norm $D$. Therefore, if $\data$ implies lower bounds on $\| I \|_1$ then these bounds can be used as lower bounds on $D$. Indeed the derivation of lower bounds on $\| I \|_1$ using $\data$ is straightforward: Consider data $\data \in \mathbb{R}^{X \times Z}$ generated by quantum states $(\rho_x)_{x=1}^X$ and a single POVM $(E_{z})_{z=1}^Z$ on $\mathbb{C}^d$, i.e., $\data_{xz} = \tr\bigl( \rho_x E_z \bigr)$. For each $z \in [Z]$ let $c^*_z := \max\{ \data_{xz} \}_{x=1}^X$ and $x' := \text{argmax}\{ \data_{xz} \}_{x=1}^X$. Then,
\[
	c^*_z = \tr\bigl( \rho_{x'} E_z \bigr) \leq \| \rho_{x'} \| \| E_{z} \|_1
\]
and therefore, $\| E_{z} \|_1 \geq c^*_z$ because $\| \rho \| \leq 1$ for all quantum states $\rho$. Here, $\| \cdot \|$ denotes the operator norm (i.e.~largest singular value) and $\|\cdot \|_1$ the trace norm (i.e.~sum of singular values). It follows that for all $z \in [Z]$, $\tr(E_{z}) \geq c^*_z$ and consequently,
\[
	\tr\Bigl( \sum_{z=1}^Z E_z \Bigr) \geq \sum_{z=1}^Z c^*_z.
\]
By assumption $\bigl( E_z \bigr)_{z=1}^Z$ is a POVM. So in particular, $\sum_{z=1}^Z E_z = \id_d$ and therefore,
\beq\label{general.dim.lb}
	d = \tr(\id_d) = \tr\Bigl( \sum_{z=1}^Z E_z \Bigr) \geq \sum_{z=1}^Z c^*_z.
\eeq
This is a lower bound on the dimension of any quantum model describing~$\data$. It has been found independently in~\cite{lee2014some} (Theorem~24). 

It is straightforward to see how the dimension lower bound $l(\data) := \sum_{z=1}^Z c^*_z$ reacts to noise in $\data$. Assume that $\data'$ is a noisy approximation of $\data$ in the sense that $\| \data - \data' \|_{\infty} \leq \varepsilon$. Here, $\| \data - \data' \|$ denotes the maximum norm associated to the vectorization of $\data - \data'$. Then,
\beq\bal\nn
	&| l(\data) - l(\data') |\\
	&\leq		 \sum_{z=1}^Z \Bigl| \, \max\{ \data_{xz} \}_{x=1}^X - \max\{ \data'_{xz} \}_{x=1}^X \, \Bigr|  \\
	&=		 \sum_{z=1}^Z \max\Bigl\{ \max\{ \data_{xz} \}_{x=1}^X - \max\{ \data'_{xz} \}_{x=1}^X,\\ & \ \ \ \ \ \ \ \ \ \ \  \max\{ \data'_{xz} \}_{x=1}^X - \max\{ \data_{xz} \}_{x=1}^X \Bigr\}\\
	&\leq		 \sum_{z=1}^Z \max\Bigl\{ \max\{ \data'_{xz} + \varepsilon \}_{x=1}^X - \max\{ \data'_{xz} \}_{x=1}^X,\\ & \ \ \ \ \ \ \ \ \ \ \  \max\{ \data_{xz} + \varepsilon \}_{x=1}^X - \max\{ \data_{xz} \}_{x=1}^X \Bigr\}	\\
	&=		Z\varepsilon.
\eal\eeq
In other words, if we measure $\data'$ with accuracy $\| \data - \data' \|_{\infty} \leq \varepsilon$, then we know that $d \geq l(\data') - Z\varepsilon$.

As an example, we consider data generated by $E_z = \ketbra{z}$, $\rho_z = \ketbra{\pi(z)}$ with $z \in [D]$, $\pi \in S_D$ a permutation and with $(\ket{z})_{z=1}^D$ being an orthonormal basis in $\mathbb{C}^D$. It follows that for all $z \in [D]$, $c^*_z = 1$. By~\eqref{general.dim.lb}, $d \geq D$, i.e., the considered quantum system generating the considered data cannot be compressed into a lower-dimensional quantum system. This example generalizes in obvious manners: if for each measurement outcome there exists a state that can be measured with approximate certainty, then roughly, the considered quantum system cannot be compressed into another quantum system whose dimension exceeds $Z$. 

We conclude that general $D$-dimensional quantum models cannot be
compressed into $d$-dimensional quantum models with $d \ll
D$. However, this still leaves room for the existence of compression
schemes that can compress specific classes of models. This is what we are going to explore next.

\section{Compression of positive semidefinite factorization}\label{Sect:Compression.of.positive.semidefinite.factorization}

Let $\mathcal{M} \in \mathbb{R}^{N \times M}_+$, let $J:=N+M$ and let
$S^+(\mathbb{C}^D)$ denote the set of psd matrices on
$\mathbb{C}^D$. The psd matrices $\bigl( A_n \bigr)_{n=1}^N$ and
$\bigl( B_m \bigr)_{m=1}^M$ from $S^+(\mathbb{C}^D)$ are said to
provide a $D$-dimensional psd factorization of $\mathcal{M}$ if for
all entries $\mathcal{M}_{nm} = \tr\bigl( A_n B_m \bigr)$. The {\em
  psd-rank} of $\mathcal{M}$ is defined to be the dimension of the
lowest-dimensional psd factorization of $\mathcal{M}$.   See
\cite{fawzi2014positive} for a recent review of psd-rank.

\subsection{The compression scheme}

Assume $\bigl( A_n \bigr)_{n=1}^N$ and $\bigl( B_m \bigr)_{m=1}^M$
form a $D$-dimensional psd factorization of $\mathcal{M}$ and set $M_n
= A_n$ for $1 \leq n \leq N$ and $M_{N+n} = B_n$ for $1 \leq n \leq
M$. Now and in the remainder we assume that $\Pi = G / \sqrt{2d} \in
\mathbb{C}^{d \times D}$ with $G_{ij} = S_{ij} + iT_{ij}$ where
$S_{ij}, T_{ij} \sim \mathcal{N}(0,1)$ independent and identically
distributed (iid).   This normalization is chosen so that $\E[\Pi]=0$
and $\E[\Pi^\dag \Pi] = I_D$.
In this section we show that for all $n,m \in [J]$, the map
\beq\label{Eq:def.of.psd.fact.comp}
	M_j \mapsto  \Pi M_j \Pi^\dag
\eeq
has the property that
\beq
	\tr\bigl( \Pi M_n \Pi^\dag \Pi M_m \Pi^\dag \bigr)
\approx \tr\bigl( M_n M_m \bigr).
\label{eq:compression.works}\eeq
Thus, it approximately preserves all of the Gram matrix (i.e., not just $\mathcal{M}$) associated to the considered psd factorization. Since $d < D$, the map~\eqref{Eq:def.of.psd.fact.comp} can be regarded as a \emph{compression} of the psd factorization we started with. The compression errors are additive. The additive error associated to $\tr(A_n B_m)$ scales with $\tr\bigl( A_n \bigr) \tr\bigl( B_m \bigr)$.

This proof is divided into two pieces.  First we argue that
\eq{compression.works} holds for rank-1 matrices $M_m,M_n$ in
section~\ref{Sect:bounds.on.Q.to.Q.compression} and then we extend to
general psd matrices $M_m,M_n$ in section~\ref{sect:psd.compression}.

\subsection{ Compression errors for vectors}
\label{Sect:bounds.on.Q.to.Q.compression}

The errors of the compression of pure states correspond to the famous
Johnson-Lindenstrauss Lemma, which has the following formulation for
complex vector spaces:

\begin{theorem}[complex Johnson-Lindenstrauss~\cite{JL84}]\label{thm:complex.JL.Lemma}
	Assume that $\Pi = G / \sqrt{2d} \in \mathbb{C}^{d \times D}$
        with $G_{ij} = A_{ij} + iB_{ij}$ where $A_{ij}, B_{ij} \sim
        \mathcal{N}(0,1)$ iid. Let ${v}_1, ..., {v}_S$ be
        arbitrary vectors in $\mathbb{C}^D$ and $\varepsilon \in
        (0,1)$. Then,
\be
\prob{\forall i\; {\|\Pi v_i\|_2} \in 
[1-\eps, 1+\eps]{\| v_i\|_2}} \geq 1 - 2S e^{d\eps^2/8}
\label{eq:length-concentration}\ee
\end{theorem}

The proof is simple enough that we reproduce it here.
\begin{proof}
	Fix a particular $v := v_i$.  
	By linearity of $\Pi$, we can assume without loss of generality that $\| v \|_2 =1$.
	Let $U$ denote a unitary matrix with the property $v = Ue_1$, where $e_1$ denotes the vector with a one in the first position and zeroes elsewhere. By unitary invariance of the complex Gaussian measure on $\bbC^D$, $\Pi\sim \Pi U$ (i.e. the two random variables are identically distributed) and so
	\be \|\Pi v\|_2^2 = \| \Pi Ue_1\|_2^2 \sim \|\Pi e_1\|_2^2 = \sum_{j=1}^d \frac{A_{1j}^2 + B_{1j}^2}{2d}
	\ee
	This average of the square of $2d$ standard normal random variables is known as the $\chi^2$ distribution, and its concentration of measure properties are standard.  Indeed, by Corollary~5.5 in~\cite{alexander2002course},
	\beq\bal\nn
		\prob{ \| \Pi v \|_2^2 \geq 1+2\varepsilon   } 
		&\leq		e^{ -\frac{1}{2} \varepsilon^2 d }\\
		\prob{ \| \Pi v \|_2^2 \leq 1-\varepsilon   } 
		&\leq		e^{ -\frac{1}{2} \varepsilon^2 d }	
	\eal\eeq
	so that
	\[
		\prob{ \| \Pi z \|^2_2 \in [1-\varepsilon,1+\varepsilon] } \geq 1 - 2 e^{-\frac{1}{8} \varepsilon^2 d}.
	\]
Eq.~\eq{length-concentration} follows by the union bound:
	\beq\bal\nn
		&\prob{ \bigcup_{i \in [S]} \Bigl\{ \| \Pi v_i \|^2_2 \not\in [(1-\varepsilon) \| v_i \|^2,(1+\varepsilon) \| v_i \|^2] \Bigr\} }\\
		&\leq		S \cdot 2 e^{-\frac{1}{8} \varepsilon^2 d} 
	\eal\eeq

\end{proof}

The usual Johnson-Lindenstrauss Lemma shows that $\Pi$ preserves (with high probability) not only the lengths of a collection of vectors, but also their pairwise distances.   In fact, we will demand slightly more.

\begin{corollary}\label{cor:JL}
Let $v_1,\ldots,v_S$, $\Pi$, $\eps$ be as in \thmref{complex.JL.Lemma}.  Then with probabilty $\geq 1 - 4S^2 e^{-d\eps^2/8}$ we have

\begsub{JL.polar}
\| \Pi v_i\|_2 & \in [1-\eps,1+\eps] \| v_i \|_2 \\
\| \Pi (v_i + v_j) \|_2 & \in  [1-\eps,1+\eps] \| v_i + v_j \|_2  \\
\| \Pi (v_i + i v_j) \|_2 & \in  [1-\eps,1+\eps] \| v_i + i v_j \|_2  \\
\| \Pi (v_i - v_j) \|_2 & \in  [1-\eps,1+\eps] \| v_i - v_j \|_2  \\
\| \Pi (v_i - i v_j) \|_2 & \in  [1-\eps,1+\eps] \| v_i - i v_j \|_2  
\endsub
\end{corollary}

\begin{proof}
Apply \thmref{complex.JL.Lemma} to the $S+ 4 \binom{S}{2} \leq 2S^2$ vectors $\{v_i\}_{i\in [S]} \cup \{v_i \pm v_j, v_i \pm i v_j\}_{1\leq i<j\leq S}$.
\end{proof}

Matrices $\Pi$ satisfying \eq{JL.polar} for some set of vectors $\{v_1,\ldots,v_S\}$ are said to be $\varepsilon$-JL matrices.
\subsection{Compression errors for positive semidefinite matrices}
\label{sect:psd.compression}

  In this section we extend \thmref{complex.JL.Lemma} to show that general psd matrices have their inner products approximately preserved by a random compression map.

First we show that $\eps$-JL matrices also approximately preserve inner products between rank-1 psd matrices (recall that $\tr(\ketbra{v}\ketbra{w}) = | \braket{v}{w} |^2$).
\begin{lemma}\label{lem:inner-product}
If $\Pi$ is an $\eps$-JL matrix for the set $\{v_1,\ldots,v_S\}$ (i.e. satisfies \eq{JL.polar}), then
\be 
\left| |\bra{v_i} \Pi^\dag \Pi \ket{v_j}|^2 - |\braket{v_i}{v_j}|^2\right |
\leq 192 \eps \|v_i\|_2^2 \|v_j\|_2^2. \label{eq:IP-close}\ee
\end{lemma}

While a direct analysis of the Gaussian probability distribution would
yield a sharper constant, our approach highlights the
fact that it is only the $\eps$-JL property that is needed. 

\begin{proof}
Since \eq{IP-close} is homogenous, we can assume that $v_i,v_j$ are unit vectors.
By the polarization identity,
\beq
	 \bra{v_i} \Pi^{\dag} \Pi \ket{v_j}
	 =	\frac{1}{4} \Bigl( x_1 - x_2 + i x_3 - i x_4) \Bigr)
\eeq
with
\bas	x_1
	&=	\bigl\| \Pi (v_i + v_j) \bigr\|_2^2 
	& x_2
	&=	\bigl\| \Pi (v_i - v_j) \bigr\|_2^2\\
	x_3
	&=	\bigl\| \Pi (v_i + iv_j) \bigr\|_2^2
	& x_4
	&=	\bigl\| \Pi (v_i - iv_j) \bigr\|_2^2\\
\eas
Thus
\be |\bra{v_i} \Pi^\dag \Pi \ket{v_j}|^2  = 
\frac{1}{16} \Bigl( (x_1 - x_2)^2 + (x_3 - x_4)^2 \Bigr) =: f(\vec x)
\ee

Define as well
\begin{align*}
	y_1
	&=	\bigl\| v_i + v_j \bigr\|_2^2 
	& y_2
	&=	\bigl\| v_i - v_j \bigr\|_2^2\\
	y_3
	&=	\bigl\| v_i + iv_j \bigr\|_2^2
	& y_4
	&=	\bigl\| v_i - iv_j \bigr\|_2^2\\
\end{align*}
Observe that each $y_i \leq 4$ and that
\beq \braket{v_i}{v_j} =
\frac{1}{4}
\Bigl( y_1 - y_2 + i y_3 - i y_4) \Bigr)
\eeq

By the $\varepsilon$-JL property of $\Pi$, we have that for $j\in [4]$, $x_j \in [(1-\eps)^2,(1+\eps)^2]y_j$.
Set
\bas
	K :=& [(1-\eps)^2y_1, (1+\eps)^2y_1] \times 
	[(1-\eps)^2y_2, (1+\eps)^2y_2] \times  \\ &
	[(1-\eps)^2y_3, (1+\eps)^2y_3] \times 
	[(1-\eps)^2y_4, (1+\eps)^2y_4].
\eas
We are going to apply
\beq\label{Schrankensatz}
	| f(\vec{x}) - f(\vec{y}) | \leq \| \nabla f \|_{K} \| \vec{x} - \vec{y} \|_{\infty}
\eeq
where $\| \nabla f \|_{K} = \max\bigl\{  \| \nabla f\|_{1} | \; \vec{x} \in K  \bigr\}$. By definition of $f(\vec{x})$,
\beq\begin{split}\label{Eq:temp.bound.on.norm.of.gradient.f}
	\| \nabla f \|_{K}
	&=		\frac{1}{4} \max_{\vec{x} \in K} \bigl\{  | x_{1} - x_{2} | + | x_{3} - x_{4} | \bigl\}\\
	&\leq		\frac{1}{4} \max_{\vec{x} \in K} \bigl\{  |x_{1}| + |x_{2}| + |x_{3}| + |x_{4}| \bigl\} \\
	&\leq		\frac{(1+\eps)^2}{4}  \bigl\{  |y_{1}| + |y_{2}| + |y_{3}| + |y_{4}| \bigl\}\\
	&\leq		4 (1+\eps)^2 \leq 16,
\end{split}\eeq
using $\eps\leq 1$ in the last step.
We also bound
\beq\begin{split}\label{Eq:temp.bound.on.infinity.norm.of.x.minus.y.in.K}
	\| \vec{x} - \vec{y} \|_{\infty} 
	&=	\max_j |x_j - y_j| \\
	&\leq	 (2\eps+\eps^2)\max_j y_j \\
	& \leq 12 \eps.
\end{split}\eeq
Thus by Eq.~\eqref{Schrankensatz},
\beq
	| f(\vec{x}) - f(\vec{y}) |
	\leq	192 \eps
\eeq
\end{proof}

Armed with this fact, we can prove \eq{compression.works}. 

\begin{reptheorem}{thm:comp.of.psd.factorizations}
\ThmCompPsdFact{second}
\end{reptheorem}

\begin{proof}
For each $i \in [J]$, let the eigendecomposition of $M_i$ be
\[
	M_i = \sum_{a=1}^D \lambda^i_a \ketbra{\psi^i_a} . 
\]
Observe that $\|M_i\|_1 = \sum_{a=1}^D \lambda^i_a$.

Choose $\Pi\in \bbC^{d\times D}$ according to a Gaussian distribution as in \thmref{complex.JL.Lemma} (i.e.~such that $\E[\Pi]=0$ and $\E[\Pi^\dag \Pi]=I_D$).  
By \corref{JL}, $\Pi$ is an $\eps$-JL matrix for the vectors $\{\ket{\psi^i_a}\}$ with probability 
\be \geq 1 - 4 J^2 D^2 e^{-d\eps^2/8}.\ee
By \eq{d-for-any-M-second}, this is $>0$, thus implying that there exists some $\Pi$ with this property.  Fix this choice of $\Pi$ for the rest of the proof.

Define 
\be M_i' := \Pi M_i \Pi^\dag.\ee
We now compute
\bas
\tr (M_i' M_j') &= 
\tr (\Pi M_i \Pi^\dag \Pi M_j\Pi^\dag) \\
& = \sum_{a,b = 1}^D \lambda^i_a \lambda^j_b |\bra{\psi^i_a}\Pi^\dag \Pi \ket{\psi^j_b}|^2
\eas
By \lemref{inner-product}, we have
\be \left | |\bra{\psi^i_a}\Pi^\dag \Pi \ket{\psi^j_b}|^2 - 
|\braket{\psi_a^i}{\psi_b^j}|^2 \right| \leq 192 \eps, \ee
for each $i,j,a,b$.
Thus
\ba
|\tr (M_i' M_j') - \tr(M_iM_j)| 
& \leq 
\sum_{a,b = 1}^D \lambda^i_a \lambda^j_b 192 \eps
\nn \\ & 
= \|M_i\|_1 \|M_j\|_1 192 \eps
\ea
\end{proof}

We remark that this proof would work equally well with general matrices, with the error terms scaling as $\|M_i\|_1 \|M_j\|_1$ instead of $\tr (M_i) \tr (M_j)$.  However, if we are willing to abandon the psd condition, a much tighter bound is possible simply by treating $\{M_i\}$ as vectors in $\bbC^{d^2}$ and using \thmref{complex.JL.Lemma}.  This would result in errors proportional instead to $\|M_i-M_j\|_2^2$ (or optionally $\|M_i\|_2 \|M_j\|_2$), both of which can be smaller than $\|M_i\|_1 \|M_j\|_1$ by as much as $D$. 

More generally, one might ask whether a nonlinear compression map can remove this dependence on the matrix norm. However, the absolute error cannot be invariant under $M_n \mapsto \lambda M_n$ ($\lambda > 0$). This can be understood as follows. Assume that $\hat{\mathcal{C}}$ is a compression with the property that there exist functions $c(\varepsilon)$ and $d(\varepsilon)$ such that 
\[
	\tr(\hat{\mathcal{C}}(M_{n})\hat{\mathcal{C}}(M_{m})) = c(\varepsilon) \tr(M_{n}M_{m})  + d(\varepsilon)
\]
for all positive semidefinite matrices $M_n, M_m$ and which are invariant under $M_n \mapsto \lambda M_n$. Then, 
\[
	\tr(\hat{\mathcal{C}}( \lambda \ketbra{n} )\hat{\mathcal{C}}(\lambda \ketbra{m})) = c(\varepsilon) \lambda^2 \delta_{nm}  + d(\varepsilon) 
\]
where $(\ket{n})_n$ denotes the canonical basis in $\mathbb{C}^D$. Set $\eta_n(\lambda) := \sqrt{c(\varepsilon) \lambda^2   + d(\varepsilon)}$. It follows that 
\[
	\left\| \frac{1}{\eta_n} \hat{\mathcal{C}}( \lambda \ketbra{n} ) \right\|_2 = 1 
\]
and for $m \neq n$
\[
	\tr\left( \frac{1}{\eta_n} \hat{\mathcal{C}}( \lambda \ketbra{n} ) \frac{1}{\eta_m} \hat{\mathcal{C}}(\lambda \ketbra{m})\right) = \frac{d(\varepsilon)}{\eta_n \eta_m} \rightarrow 0
\]
as $\lambda \rightarrow \infty$ because, by assumption, $d(\varepsilon)$ is invariant under $M_n \mapsto \lambda M_n$. Assuming $D > d^2$, we would be able to construct an overcomplete orthonormal basis of positive semidefinite matrices. This is impossible and we conclude that necessarily, $d(\varepsilon) = \Omega(\lambda^2)$.

\section{Compression of quantum models with few measurement outcomes}\label{Sect:comp.of.quant.models}

The purpose of this section is the definition and the analysis of a scheme to compress quantum models. The discussion in section~\ref{Sect:incompressibility} shows that the compressibility of quantum models is limited by the number of measurement outcomes. This will be reflected in the conditions for the compression to be applicable.

Hence, given $\data$ generated by arbitrary quantum states $\rho_x$ and POVMs $\bigl( E_{yz} \bigr)_{z=1}^Z$ on $\mathbb{C}^D$, our goal is to find low-dimensional states and POVMs on $\mathbb{C}^d$ which reproduce $\data$ approximately ($d < D$). To treat the states and the POVMs all at once, we define $M_x := \rho_x$ for $x \in [X]$ and $M_{X + (y-1)Z + z} := E_{yz}$ for $y \in [Y]$ and $z \in [Z]$. So the matrices $(M_n)_{n=1}^J$ form a list of matrices describing the states and POVMs. Then, the matrix $\mathcal{G}_{nm} = \tr\bigl( M_n M_m \bigr)$ is the Gram matrix generated by all of the states and all of the measurement operators. Hence, $\mathcal{G}$ describes the Euclidean geometry of the set $\bigl( M_n \bigr)_{n=1}^J$.

The transition from the high-dimensional states and POVMs to their low-dimensional counterparts will be described in terms of a compression map $\comp: S^+(\mathbb{C}^D)^{\times J} \rightarrow S^+(\mathbb{C}^d)^{\times J}$,
\beq\nn
	\comp:  \bigl( M_n \bigr)_{n=1}^J \mapsto \comp\Bigr(\bigl( M_n \bigr)_{n=1}^J\Bigr)  = \bigl( \comp_n(M_n) \bigr)_{n=1}^J.
\eeq
We will see that the proposed compression scheme $\comp$ approximately preserves not only the inner products $\mathcal{M}_{n,m}$ but also the entire Gram matrix $\mathcal{G}$. The compression errors are additive. They scale with the trace norm of the measurement operators. In the remainder we are going to suppress the index $n$ of $\comp_n(\cdot)$.




\subsection{The compression scheme}
\label{sec:description.of.Q.to.Q.compression.scheme}

Let $\varepsilon > 0$ and set $J = X + YZ$. Assume that $\Pi = G / \sqrt{2d} \in \mathbb{C}^{d \times D}$ with $G_{ij} = A_{ij} + iB_{ij}$ where $A_{ij}, B_{ij} \sim \mathcal{N}(0,1)$ iid. Set
\beq\label{Eq:def.comp.of.states}
	\comp(\rho_x) = \frac{1}{\tr(\Pi \rho_x \Pi^\dag)} \; \Pi \rho_x \Pi^\dag
\eeq
for all $x \in [X]$,
\beq\label{eq:def.comp.of.POVMs.part.1}
	\comp(E_{yz}) = \frac{1}{1+\varepsilon} \; \Pi E_{yz} \Pi^\dag
\eeq
for all $y \in [Y]$ and $z \in [Z-1]$ and set 
\beq\label{Eq:def.comp.of.POVMs.part.2}
	\comp(E_{yZ}) = \id - \sum_{z=1}^{Z-1} \comp(E_{yz})
\eeq
for all $y \in [Y]$. Note that the map $\mathcal{C}$ is non-linear. Obviously, $\comp(\rho_x)$ are valid quantum states. In section~\ref{Sect:Verification.that.compression.of.POVM.is.psd.and.normalized}, we are going to show that with non-vanishing probability, the matrices $\bigl( \comp(E_{yz}) \bigr)_{z=1}^Z$ form a valid POVM and therefore, the image of $\comp$ gives valid quantum states and measurements on $\mathbb{C}^d$. 


\subsection{Probability that compression of a POVM is psd and normalized}\label{Sect:Verification.that.compression.of.POVM.is.psd.and.normalized}

It is not {\em a priori} obvious that 
$\bigl( \comp(E_{yz}) \bigr)_{z=1}^Z$ is indeed a valid POVM.  While $\comp(E_{yz})\geq 0$ for $z<Z$ and $\sum_{z=1}^Z \comp(E_{yz}) = I$ hold automatically, it is not always true that 
\be \comp(E_{yZ}) \geq 0
\quad\Leftrightarrow\quad
\Pi \sum_{z=1}^{Z-1} E_{yz} \Pi^\dag \leq (1+\eps)I
 \label{eq:EyZ-psd}\ee
However, for appropriate choice of $\eps$, we will see that \eq{EyZ-psd} holds with high probability.  Let $P$ be a projector onto the support of $\sum_{z=1}^{Z-1} E_{yz}$.  Let $r = \rank P = \rank \sum_{z=1}^{Z-1} E_{yz}$.  Then $\|\Pi P \Pi^\dag\| = \sigma_{\max}(\Pi P)^2$, where $\sigma_{\max}(\cdot)$ denotes the largest singular value.  Because of the unitary invariance of the Gaussian measure, $A := \Pi P$ is distributed identically to a $d\times r$ matrix of complex i.i.d. Gaussians with mean 0 and variance $1/d$. 
Specifically observe that $\E[A A^\dag] = I_r$.  

We now appeal to a standard result in random matrix theory.
\begin{theorem}[\cite{Haagerup03}]\label{thm:bound.on.sigma.max}
	Let $A$ be a $r \times d$ Gaussian matrix with
iid entries satisfying $\E[A_{ij}]=0, \E[|A_{ij}|^2]=1/d$. Then for $0\leq t\leq d/2$ and $0\leq \delta \leq 2$ we have
\ba
\E[\exp(t \|A A^\dag \|)]  &\leq
d \exp \L(t(1+\sqrt{r/d})^2 + \frac{t^2}{d}(1+r/d)\R) 
\label{eq:mgf-bound}\\
\prob{\|AA^\dagger\| &\geq (1+\sqrt{r/d})^2 + \delta}
 \leq d e^{-d\delta^2/8} \label{eq:tail-bound}
\ea
\end{theorem}

\begin{proof}
\eq{mgf-bound} is Lemma 7.2 of \cite{Haagerup03} and \eq{tail-bound} follows by setting $t=d\delta/4$, and using the bound $\prob{\|AA^\dag\| \geq \lambda} \leq \E[\exp(t\|AA^\dag\|)]e^{-t\lambda}$ and the fact that $r\leq d$.
\end{proof}

This implies that our dimension-reduction strategy does not significantly blow up the POVM normalization as long as the rank of the measurement operators is much less than $d$. 
\begin{lemma}\label{lem:control-norm}
	Let $\Pi = G / \sqrt{2d} \in \mathbb{C}^{d \times D}$ with $G_{ij} = A_{ij} + iB_{ij}$ where $A_{ij}, B_{ij} \sim \mathcal{N}(0,1)$ iid. Assume $E \in \mathbb{C}^{D \times D}$ psd with $\| E \| \leq 1$. Then,
	\beq\label{eq:prob.bound.for.op.norm.of.E}
		\prob{ \| \Pi E \Pi^\dag \| \leq 1 + \eps  } \geq 1 - de^{-\frac{\eps^2}{32} d}
	\eeq
	if
	\beq\label{eq:d-beats-r}
		d \geq \frac{32}{\eps^2}  \, \rank(E).
	\eeq
\end{lemma}

\begin{proof}
Let $r =\rank E$.  Since $E\leq I$, then $E\leq P$ for some rank-$r$ projector $P$.    Then, according to the discussion earlier in this section, $\|\Pi E \Pi^\dag \| \leq \|A A^\dag\|$ for some matrix $A$ satisfying the conditions of \thmref{bound.on.sigma.max}.  If \eq{d-beats-r} holds then with probability $\geq 1- d e^{-d\eps^2/32}$ we have
\be  
\|\Pi E \Pi^\dag\|  \leq \|AA^\dag \|
\leq (1 + \sqrt{r/d})^2 + \eps/2.\ee
From \eq{d-beats-r} and $\eps\leq 1/2$ we have 
\beq\bal
(1 + \sqrt{r/d})^2 &\leq (1 + \eps/\sqrt{32})^2 \\
& \leq 1 + \eps\L(\frac{2}{\sqrt{32}} + \frac{1}{64}\R)
\leq 1 + 0.37\eps.\eal\eeq
Together this implies that $\|\Pi E \Pi^\dag\| \leq 1 + \eps$.
\end{proof}

We now apply this to dimension reduction.  For each $y \in [Y]$ set
\beq\label{eq:event.E.y}
	\mathcal{E}_{y} = \{ \Pi \; | \; \| \Pi \Bigl( \sum_{z \in [Z-1]} E_{yz} \Bigr) \Pi^\dag \| \leq 1+\varepsilon \}.
\eeq
By Lemma~\ref{lem:control-norm} and the union bound
\beq\label{prob.lb.for.event.E.v}
	\prob{ \bigcap_{y\in [Y]} \mathcal{E}_{y} } \geq  1 - Y e^{-\frac{\varepsilon^2}{32} d}
\eeq
if
\beq\label{condition.for.application.lb.for.Ev}
	d > \frac{32}{\varepsilon^2}  \, \rank\Bigl( \sum_{z=1}^{Z-1} E_{yz} \Bigr).
\eeq

\subsection{Tail bounds on compression errors}\label{Sect:tail.bounds.on.Q.to.Q.compression}

We use \eq{matrix-compression-error-second} to separately list the compression errors for inner products $\tr\bigl[ \rho_x \rho_{x'} \bigr]$, $\tr\bigl[ \rho_x E_{yz} \bigr]$ and $\tr\bigl[ E_{yz} E_{y'z'} \bigr]$.  Let $\cJ$ denote the event where $\Pi$ satisfies the $\eps$-JL property for the eigenstates of all the $\rho_x$ and $E_{yz}$.

This event implies the following bounds.
\beq\begin{aligned}\label{Eq:temp.f8e5sej2}
	\tr\bigl[ \Pi \rho_x \Pi^\dag \bigr] &= \sum_{j=1}^D p_j \tr\bigl[ \Pi \ketbra{\psi^x_j} \Pi^\dag \bigr] =  \sum_{j=1}^D p_j \| \Pi \ket{\psi^x_j} \|_2^2\\ &\in [(1-\varepsilon)^2, (1+\varepsilon)^2].
\end{aligned}\eeq

\emph{State-state error.} For $n,m \in [X]$,~\eq{matrix-compression-error-second} implies (using~\eqref{Eq:temp.f8e5sej2})
\beq\begin{split}\label{Eq:final.bounds.on.compressed.state.state.ips}
	& \frac{1}{(1+\varepsilon)^4} \Bigl( \tr(\rho_{x}\rho_{x'})  -  192 \eps \Bigr)\\
	&\leq	\tr\bigl(\comp(\rho_{x}) \comp(\rho_{x'}) \bigr)  \\
	&\leq	  \frac{1}{(1-\varepsilon)^4} \Bigl(\tr(\rho_{x}\rho_{x'})  +  192 \eps \Bigr)
\end{split}\eeq

\emph{State-measurement error.} Let $z \in [Z-1]$. By~\eq{matrix-compression-error-second},
\beq\begin{split}\label{Eq:final.bounds.on.compressed.state.meas.ips.for.k.from.to.1.to.K-1}
	& \frac{1}{(1+\varepsilon)^3} \Bigl( \tr(\rho_{x}E_{yz})  - \tr(E_{yz}) 192 \eps \Bigr)\\
	&\leq	\tr\bigl(\comp(\rho_{x}) \comp(E_{yz}) \bigr)  \\
	&\leq	  \frac{1}{(1-\varepsilon)^2(1+\varepsilon)} \Bigl(\tr(\rho_{x}E_{yz})  + \tr(E_{yz}) 192 \eps \Bigr)
\end{split}\eeq
Let $z=Z$. By~\eqref{Eq:def.comp.of.states}, \eqref{eq:def.comp.of.POVMs.part.1} and~\eqref{Eq:def.comp.of.POVMs.part.2},
\beq\begin{aligned}\label{Eq:final.bounds.on.compressed.state.meas.ips.for.k.equal.to.K}
	& 	\tr\bigl[ \mathcal{C}(\rho_x) \mathcal{C}(E_{yZ}) \bigr]\\
	&=	\tr\Bigl[ \frac{\Pi \rho_x \Pi^\dag}{\tr[ \Pi \rho_x \Pi^\dag ]}  \Bigl( \id - \Pi \Bigl( \sum_{z=1}^{Z-1} E_{yz} \Bigr)\Pi^\dag \Bigr) \Bigr]\\
	&=	1 - \Bigl( \sum_{z=1}^{Z-1} \tr\bigl[ \mathcal{C}(\rho_x) \mathcal{C}(E_{yz}) \bigr] \Bigr)
\end{aligned}\eeq
All the summands in~\eqref{Eq:final.bounds.on.compressed.state.meas.ips.for.k.equal.to.K} can be bounded by~\eqref{Eq:final.bounds.on.compressed.state.meas.ips.for.k.from.to.1.to.K-1}.

\emph{Measurement-measurement error.} Let $z,z' \in [Z-1]$. By~\eq{matrix-compression-error-second},
\beq\begin{split}\label{Eq:final.bounds.on.compressed.meas.meas.ips.for.k.from.to.1.to.K-1}
	& \frac{1}{(1+\varepsilon)^2} \Bigl( \tr(E_{yz}E_{y'z'})  - \tr(E_{yz}) \tr(E_{y'z'}) 192 \eps \Bigr)\\
	&\leq	\tr\bigl(\comp(\tr(E_{vk})) \comp(E_{y'z'}) \bigr)  \\
	&\leq	  \frac{1}{(1+\varepsilon)^2} \Bigl(\tr(E_{vk} E_{y'z'})  + \tr(E_{vk}) \tr(E_{y'z'}) 192 \eps \Bigr)
\end{split}\eeq
For $z=Z$ we get bounds similar to~\eqref{Eq:final.bounds.on.compressed.state.meas.ips.for.k.equal.to.K}.

In each case the denominators can be absorbed by rounding up $192\eps$ to $200\eps$.  This is a straightforward calculation in the case when $\eps\leq 1/200$ and when $\eps>1/200$ the error bound is vacuously true.

Here, we are able to avoid the no-go Theorem
from~\cite{harrow2011limitations} (already for the factorization of
the identity matrix) because $\mathcal{C}$ is not
completely positive, or even linear.  Indeed, generically each row of $\Pi$ will have norm $\sqrt{D/d}$ and so we will have $\|\Pi\| \approx \sqrt{D/d}$ with high probability.  Thus there will almost always exist matrices $E$ such that $\|\Pi E \Pi^\dag \| \gg \|E\|$.  However, our compression scheme can work because $\|\Pi E \Pi^\dag\|$ is small for most $E$ (or more precisely, for all $E$ it is small for most $\Pi$).  This feature of being able to achieve something for a small number of vectors that is impossible to extend to all vectors can be seen already in the original Johnson-Lindenstrauss Lemma.

\subsection{Proof of \thmref{comp.of.q.models}}\label{Sect:q.models:union.bound}

\begin{proof}
Recall the definitions of the events $\mathcal{E}_{y}$ and $\mathcal{J}$ from~\eqref{eq:event.E.y} and Section~\ref{Sect:tail.bounds.on.Q.to.Q.compression}, respectively.  By~\eqref{prob.lb.for.event.E.v} and \corref{JL} the probability that these all hold simultaneously is
\be
\prob{ \mathcal{J} \cap  \bigcap_{y \in [Y]} \mathcal{E}_{y} } \geq 1 - Y e^{-\frac{\eps^2}{32}d}
 - 4J^2D^2 e^{-\frac{\eps^2}{8}d}.
\ee
If $d$ satisfies \eq{d.big.enough} and $J,D\geq 2$, then this is 
$$ \geq 1 - \frac{Y}{JD} - \frac{4J^2D^2}{(4JD)^2} \geq \frac 1 4.$$

Fix a $\Pi$ for which $\cJ$ and $\bigcap_y \cE_y$ hold.  Let $\comp$ be the corresponding compression scheme described in Section~\ref{sec:description.of.Q.to.Q.compression.scheme}.  Then by the arguments in Sections~\ref{Sect:Verification.that.compression.of.POVM.is.psd.and.normalized} and \ref{Sect:tail.bounds.on.Q.to.Q.compression}, $\cC$ satisfies the bounds in \thmref{comp.of.q.models}.

\end{proof}

\subsection{Implications}\label{Sect:Implications}

\subsubsection{Quantum tomography}\label{Sect:implications.for.quantum.tomo}

%
%
%

Every empirical observation can be captured in the form of data tables $\data$,   
\beq\label{def.data.D}
	\mathcal{D} =   \left( \begin{array}{ccccccc}  f_{1|11} & \cdots & f_{Z|11} & \cdots &   f_{1|1Y} & \cdots & f_{Z|1Y}  \\  f_{1|21} & \cdots & f_{Z|21} & \cdots &   f_{1|2Y} & \cdots & f_{Z|2Y} \\  \vdots &   & \vdots &   &   \vdots &   & \vdots  \\   f_{1|X1} & \cdots & f_{Z|X1} & \cdots &   f_{1|XY} & \cdots & f_{Z|XY}  \end{array} \right) .
\eeq
Here, $f_{z|xy}$ describes the empirical relative frequency for measuring the measurement outcome $z$ given that the considered system is in a state labelled by $x$ and given that we have performed the measurement $y$. If the observed system is quantum mechanical, then the states are described in terms of density matrices $\bigl( \rho_x \bigr)_{x=1}^X$ and each measurement $y$ is described in terms of measurement operators $\bigl( E_{yz} \bigr)_{z=1}^Z$. Note that the indices $x,y,z$ are not restricted to be scalar but can be multidimensional. For instance $x = (x_1,x_2,x_3) \in \mathbb{R}^3$ if $\rho_x$ is associated with a state preparation at spatial coordinates $(x_1,x_2,x_3)$, or $y = t \in \mathbb{R}$ if a measurement is repeated at different times $t$.

Of course, in practice, $f_{z|xy}$ only approximates $\tr\bigl( \rho_x E_{yz} \bigr)$. One way to model our uncertainty is by assuming that
\beq\label{Eq:assumed.relation.frequency.vs.probability}
	\| \bigl( f_{z|xy} - \tr(\rho_x E_{yz}) \bigr)_{xyz} \|_{\infty} \leq \delta
\eeq
for $\delta > 0$, i.e., $| f_{z|xy} - \tr(\rho_x E_{yz}) | \leq \delta$ for all $(x,y,z)$. Ideally we would be able to characterize the set $Q[\data]$ of all quantum models that are compatible with $\data$ given the uncertainty~\eqref{Eq:assumed.relation.frequency.vs.probability}. The set $Q[\data]$ is a subset of cartesian products of state space and measurement space. One practical characteristic of $Q[\data]$ is the dimension $d_{\min}$ of the lowest dimensional model in $Q[\data]$. For instance, if our guiding principle is Occam's razor (e.g., to counteract overfitting) then, to describe $\data$ theoretically, we should report a model from $Q[\data]$ whose dimension equals $d_{\min}$.

Denote by $\vec{\data}_{yz}$ the column of $\data$ corresponding to the measurement operator $E_{yz}$. The dimension $d_{\min}$ can be trivially upper bounded by $X$ because, as already pointed out in~\cite{harrigan2007representing}, the choices $\rho_x$ on $\mathbb{C}^{X}$ with
\[
	 \rho_x = \proj{x}
\]
and $E_{yz}$ on $\mathbb{C}^{X}$ with
\[
	\bigl( E_{yz} \bigr)_{i,j} = \bigl( \vec{\data}_{yz} \bigr)_i \; \delta_{i,j}
\]
(corresponding to the factorization $\data = \id \data$) is a valid model for $\data$, i.e., 
\[
	\bigl( \rho_1,..., \rho_X,E_{11},...,E_{YZ} \bigr) \in Q[\data].
\]	
To see this, recall that due to normalization within $\data$, $\sum_{z} \vec{\data}_{yz} = (1,...,1)^T$ and consequently, not only are the proposed states normalized and  psd but we also have that $\sum_z E_{yz} = \id$.

Finding non-trivial upper bounds on $d_{\min}$ is genuinely difficult. This is why here we address the following question: given a $D$-dimensional model in $Q[\data]$, can we---while respecting~\eqref{Eq:assumed.relation.frequency.vs.probability}---find a $d$-dimensional model in $Q[\data]$ with $d \ll D$? In other words, is it possible to \emph{compress} the given $D$-dimensional model into a lower-dimensional model? 

The results from sections~\ref{Sect:incompressibility} and~\ref{Sect:q.models:union.bound} show that the answer is yes and no. It is affirmative for models whose measurements have few outcomes and small trace norm (see section~\ref{Sect:q.models:union.bound}) and it is negative for measurements whose POVM elements are too long to be compressed into a compressed measurement space (see section~\ref{Sect:incompressibility}).

We already discussed an example for the incompressibility of quantum models in section~\ref{Sect:incompressibility}. Now we consider a closely related quantum model which admits exponential compression. Assume $\rho_x$ are $D$-dimensional states, $Z = 2$ and assume that for all $y \in [Y]$ we have $\tr(E_{y1}) = 1$. This example is of practical relevance because carefully calibrated measurements are often believed to be clean in the sense that their POVM elements have trace norm $\mathcal{O}(1)$ (e.g., rank-1 projectors). In this setting, the compression $\comp$ allows us to compress these states and measurements into a quantum model of dimension $d = \mathcal{O}(\log(D)/\varepsilon^2)$ while 
\beq\begin{split}\nn
\L| 
\tr\bigl(\comp(\rho_{x}) \comp(E_{y1}) \bigr)  
- \tr(\rho_{x}E_{y1})\R | \leq 200 \varepsilon
\end{split}\eeq
(recall~\eqref{Eq:final.bounds.on.compressed.state.meas.ips.for.k.from.to.1.to.K-1}). It follows that
\beq\label{fewJIJIOefjwe22}
	\Bigl\|   \Bigl( \tr(\rho_{x}E_{yz}) - \tr\bigl(\comp(\rho_{x}) \comp(E_{yz}) \bigr) \Bigr)_{xyz}   \Bigr\|_{\infty} \leq 200 \varepsilon
\eeq
because $\tr(\rho_{x}E_{y2}) = 1 - \tr(\rho_{x}E_{y1})$. Let $\data$ denote the measurement data associated to the states $\rho_x$ and the binary measurements $E_{yz}$. 

By~\eqref{Eq:assumed.relation.frequency.vs.probability}, $\bigl( \tr(\rho_{x}E_{yz}) \bigr)_{xyz}$ satisfies $\| (\tr[ \rho_x E_{yz} ])_{xyz} - \data \|_{\infty} \leq \delta$. However, the precise values $\tr[ \rho_x E_{yz} ]$ are unknown to the experimentalist. Set $\varepsilon = \delta / 200$. Then, for this choice for $\varepsilon$, the compressed states $\comp(\rho_x)$ and measurements $\comp(E_{xy})$ satisfy~\eqref{Eq:assumed.relation.frequency.vs.probability} and thus, they form a valid and exponentially smaller quantum model on $\mathbb{C}^d$ with $d = \mathcal{O}(\log(D)/\varepsilon^2)$.

\subsubsection{Limits of robust dimension witnessing}\label{sec:limits.for.dimension.witnessing}

A function $f: \data \mapsto f(\data) \in \mathbb{R}$ is a so-called quantum dimension witness~\cite{brunner2013dimension} if for some $Q_D \in \mathbb{R}$,
\beq\label{temp.fiejJeij3}
	f\bigl( (\tr[ \rho_x E_{yz} ])_{xyz} \bigr) \leq Q_D
\eeq
for all quantum states $\rho_x$ and measurements $E_{yz}$ whose dimension is $\leq D$. Dimension witnesses were introduced in the context of Bell inequalities~\cite{acin2007device}. In the prepare-and-measure scenario we are considering here, they have been studied extensively in the past years (see for instance~\cite{harrigan2007representing,wehner2008lower,gallego2010device,stark2014self}). Assume $\rho_x'$ and $E_{yz'}$ are such that
\beq\nn
	f\bigl( (\tr[ \rho'_x E'_{yz} ])_{xyz} \bigr) > Q_{k}.
\eeq
By~\eqref{temp.fiejJeij3}, the dimension $D$ of $\rho_x'$ and $E_{yz'}$ satisfies $D > k$. Hence, dimension witnesses yield lower bounds on the Hilbert space dimension. Indeed, the search for dimension witnesses is motivated by the need to certify high-dimensionality of quantum systems in a device-independent manner, i.e., by only looking at the measured data~$\data$. In the following, $f^*\bigl( (\tr[ \rho'_x E'_{yz} ])_{xyz} \bigr)$ denotes the dimension lower bound that is implied by the dimension witness $f$.

Robustness of dimension witnesses with respect to loss has been studied in~\cite{dall2012robustness}. A more general objective is the analysis of the dimension witness' robustness against more general noise $\data_{xyz} \neq \tr[ \rho'_x E'_{yz} ]$. Here, we quantify noise in the measured data $\data$ in terms of $l_{\infty}$-norm on $\mathbb{R}_+^{X \times YZ}$. One approach to define robustness of dimension witnesses is the demand that there exists $L > 0$ such that for all $\delta > 0$
\beq\bal\label{fjeiJIj342uhufh}
	| f^*\bigl( (\tr[ \rho'_x E'_{yz} ])_{xyz} \bigr) - f^*\bigl( \data \bigr) |  \leq L  \| (\tr[ \rho'_x E'_{yz} ])_{xyz} - \data \|_{\infty}.
\eal\eeq 
In other words, the function $f^*$ is Lipschitz-continuous.

Recall the binary example from section~\ref{Sect:implications.for.quantum.tomo}. There, the compressed model has dimension $\mathcal{O}(\log(D)/\varepsilon^2)$ even though the original model is $D$-dimensional. By~\eqref{fewJIJIOefjwe22} and \eqref{fjeiJIj342uhufh},
\beq\label{fejH3482}
	\Bigl| f^*\bigl( (\tr[ \rho_x E_{yz} ])_{xyz} \bigr) - f^*\Bigl( \bigl( \tr[\comp(\rho_{x}) \comp(E_{yz}) ] \bigr)_{xyz} \Bigr) \Bigr|  \leq L \cdot 200 \varepsilon
\eeq
if $f^*$ is Lipschitz-continuous. Since $f^*$ provides a lower bound on the dimension,
\[
	f^*\Bigl( \bigl( \tr[\comp(\rho_{x}) \comp(E_{yz}) ] \bigr)_{xyz} \Bigr) = \mathcal{O}(\log(D)/\varepsilon^2).
\]
From \eqref{fejH3482}, we have
\beq\label{feijiIJIJwdksk2213}
\bal
	f^*\Bigl( (\tr[ \rho_x E_{yz} ])_{xyz} \Bigr) &= \mathcal{O}(\log(D)\varepsilon^{-2} + L\eps)
\\ &	= \cO(\max(\log(D) ,L^{2/3} \log(D)^{1/3})),
\eal\eeq
where in the second step we have chosen $\eps = \min(1/2, L^{-1/3}(\log D)^{1/3})$. 

We conclude the following for experiments whose measurements have few outcomes and small trace norm: For all Lipschitz-continuous lower bounds $f^*$ there will be an exponential gap~\eqref{feijiIJIJwdksk2213} between the dimension of the underlying Hilbert space and the dimension lower bound certified by the lower bound $f^*$ and $f$ respectively.

\subsubsection{Consequences for one-way quantum communication complexity}\label{sec:quant.1.way.complexity}
Let $f: \{0,1\}^{n} \times \{0,1\}^{m} \rightarrow \{0,1\}$ be a Boolean function. Assume Alice gets the input $x \in \{0,1\}^{n}$ and Bob gets $y \in \{0,1\}^{m}$. In one-way quantum communication protocols~\cite{klauck00}, Alice is allowed to send a single quantum state $\rho_x$ to Bob. After receiving $\rho_x$, Bob tries to output $f(x,y)$. For that purpose, he chooses a POVM $(E_y, \id - E_y)$ and measures $\rho_x$. If he measures $E_y$ he sets $a=1$. Otherwise, he sets $a=0$. The one-way quantum bounded error communication complexity of $f$, denoted $Q^1(f)$, is the minimal number of qubits Alice needs to send to Bob so that $a = f(x,y)$ with probability $2/3$. 

To understand the connection to psd factorizations, let $A \in \mathbb{R}^{2^n \times 2^m}$ such that $A_{xy} = f(x,y)$. Hence, $A$ is a $0/1$-matrix describing $f$. Alice and Bob try to find the smallest $d$  such that for every $x,y$ there exists a state $\rho_x$ and a POVM element $E_y$ such that $\tr(\rho_x E_y) \geq 2/3$ if $A_{xy} = 1$ and $\tr(\rho_x E_y) \leq 1/3$ if $A_{xy} = 0$. Recalling~\eqref{def.data.D}, we note that this is equivalent to finding an approximate low-dimensional quantum model for the ($2^n \times 2^{m+1}$)-data table  
\[
	\data(A) := \bigl(  ( \vec{A}_y | \vec{1} - \vec{A}_y )  :  y \in \{ 0,1 \}^{m} \bigr)
\]
where $\vec{A}_y$ denotes the column of $A$ associated the input $y$ on Bob's side. In $\data(A)$, the two column vectors $( \vec{A}_y | \vec{1} - \vec{A}_y )$ correspond to the binary POVM 
\[
	(E_{y1}, E_{y2}) := (E_y, \id - E_y).
\] 

Let $P$ be a  $D$-dimensional protocol to solve the one-way quantum communication task associated to a specific boolean function $f$. Hence, $Q^1(f) \leq \log_2(D)$. Being a protocol, $P$ specifies mappings $x \mapsto \rho_x$ and $y \mapsto E_y$ for all inputs $x,y$.

We would like to apply Theorem~\ref{thm:comp.of.q.models} to reduce
the dimension of the states and measurements used by $P$.  This can result in significant savings whenever Bob's measurements are sufficiently unbalanced; that is, for each $y$, either $\rank E_{y1}$ or $\rank E_{y2}$ is small.  Specifically define
\be r := \max_{y \in \{0,1\}^m} \min_{z \in \{0,1\}} \rank E_{yz}.\ee
Then we will show that any bounded-error one-way protocol can be compressed to reduce the communication cost to $O(\log(nm r \log(D)))$.

\begin{corollary}\label{Cor:about.quant.1.way.complexity}
	Let $x \in \{0,1\}^n$ and $y \in \{0,1\}^m$. Assume that $x
        \mapsto \rho_x$ and $y \mapsto E_y$ is a $D$-dimensional
        protocol to solve the one-way quantum communication task
        associated to a Boolean function $f: \{0,1\}^n \times
        \{0,1\}^m \rightarrow \{0,1\}$ with error $\eps_0$. Suppose that
\be
d\geq \frac{1280000 r^2}{\eps_1^2}
\max\L( \ln(4(2^n + 2^{m+1})D), r\R)
\label{eq:d-for-CC}
\ee
Then there exists a protocol that transmits a $d$-dimensional state from Alice to Bob and achieves a worst-case error of $\leq \eps_0 + \eps_1$.
\end{corollary}

\begin{proof}
Assume WLOG that for each $y$, $\rank E_{y1} \leq \rank E_{y2}$, so that $r = \max_y \rank(E_{y1})$.   Let $\eps = \eps_1 / 200r$.  Observe that $\tr E_{y1} \leq r$ for each $y$.  We would like to apply \thmref{comp.of.q.models} to this collection of measurements.  Here $J = 2^n + 2^m\cdot 2$. 
Thus we
we obtain $d$-dimensional compressed states $\{\rho_x'\}$ and measurements $\{E_{yz}'\}$ such that
\be
\L| \tr(\rho_x' E_{yz}')  - \tr(\rho_x E_{yz})\R| \leq \eps_1
\ee
Here we also used the fact that $\tr(\rho_x' E_{y2}') = 1 - \tr(\rho_x'E_{y1}')$, so we can extend the error bounds for the $z=1$ measurements to $z=2$ measurements.
\end{proof}

\subsubsection{Quantum message identification}\label{Sect:ID}

The problem of quantum message identification (introduced by Winter in \cite{winter2004quantum}) can be thought of as a quantum generalization of the problem of testing whether two bit strings are equal.  Alice and Bob get descriptions of $D$-dimensional pure states $\tau = \ketbra{\phi}$ and $\pi = \ketbra{\theta}$, respectively.  The goal is for Bob to output a bit that equals 1 with probability close to $\tr(\tau\pi)$, which we think of primarily as guessing whether $\tau$ and $\pi$ are approximately equal or approximately orthogonal.  The problem would be trivial if Alice could send Bob the state $\tau$, but instead she is restricted to transmitting only a $d$-dimensional system.  (Winter considered a more general setting in which the parties can also use entanglement and/or shared or private randomness.)

Alice can use a channel $T: \mathcal{A}_1 \rightarrow \mathcal{A}_2$ to send quantum states to Bob. In quantum message identification, Bob needs to decide (up to some error) whether or not $\tau \approx \pi$. The map $\mathcal{E}: \tau \mapsto \mathcal{E}(\pi) \in \mathcal{A}_1^{\otimes n}$ describes Alice's encoding of her state $\pi$ before she uses the channel $T$ $n$-times to send her encoded state to Bob. Upon receiving $T^{\otimes n}\bigl( \mathcal{E}(\pi) \bigr)$, Bob applies a binary measurement $(\mathcal{D}_{\tau}, \id - \mathcal{D}_{\tau})$ to decide whether or not he should reject that $\tau \approx \pi$. Hence, $\mathcal{D}$ maps $\tau$ to a POVM element, i.e., $0 \leq \mathcal{D}_{\tau} \leq \id$. The maps $(\mathcal{E},\mathcal{D})$ provide an $(n,\lambda)$-quantum-ID code~\cite{winter2004quantum} if 
\beq\label{Winter.def.of.id.code}
	\forall \pi,\tau \ \ \bigl|  \tr(\pi\tau) - \tr\left( T^{\otimes n}\bigl( \mathcal{E}(\pi) \bigr) \mathcal{D}_{\tau} \right)  \bigr| \leq \lambda/2.
\eeq
Let now $T = \mathrm{id}: \mathbb{C}^d \rightarrow \mathbb{C}^d$ be the identity channel and therefore, fix $n=1$. Given $d$ and $\lambda$, how large can $D$ (i.e., the dimension of the message space) be? The following Theorem is the main result of~\cite{winter2004quantum}.

\begin{theorem}[Proposition~17 in~\cite{winter2004quantum}]\label{Winter.thm.about.message.id}
	For $0 < \lambda < 1$, there exists on $\mathbb{C}^d$ a quantum-ID code of error $\lambda$ such that $D = \left\lfloor d^2 \frac{(\lambda/100)^4}{4 \log(100/\lambda)} \right\rfloor$.
\end{theorem}

Equation~\eqref{Winter.def.of.id.code} can be reinterpreted as the task to find mappings $\mathcal{E}$ and $\mathcal{D}$ to ``compress" quantum states and rank-1 measurements such that the inner products between states and measurements are preserved approximately. This observation provides the link between the study of quantum message identification and the compression of quantum models. On a technical level, the results from~\cite{winter2004quantum} and our results differ in two regards. Firstly, Theorem~\ref{Winter.thm.about.message.id} does not take into account situations where both Alice and Bob get mixed quantum states as inputs. These situations are covered by our considerations. Secondly, the definition in~\eqref{Winter.def.of.id.code} demands approximate preservation of the inner product between all possible pure states whereas we consider finite families of states and measurements. For finite sets we get exponential instead of quadratic compression. For example in case of pure quantum messages, Alice and Bob agree on a family of pure quantum messages $\mathcal{F} = \{ \ketbra{\psi_j} \}_{j} \subseteq \mathbb{C}^{D\times D}$. Then, Alice gets $\pi \in \mathcal{F}$ and Bob gets $\tau \in \mathcal{F}$. Due to the purity of $\pi$ and $\tau$, the rank constraints in~\eqref{condition.for.application.lb.for.Ev} are not a bottleneck. Thus, we can compress into a Hilbert space with dimension $\mathcal{O}\bigl(\log(| \mathcal{F} |)/\varepsilon^2 \bigr)$. The compression error is controlled by~\eqref{Eq:final.bounds.on.compressed.state.state.ips}.

\section{Compression of a POVM with exponentially decaying spectrum}\label{Sect:Verification.that.compression.of.POVM.is.psd.and.normalized:incorporating.spectral.tails}

The previous discussion assumed that with the exception of a single
POVM element per measurement, all POVM elements have small rank. This
constraint might be unnatural when we are asking for the compression
of experimental states and measurements. It is more natural to expect
that experimental measurements are full rank with thin spectral
tails. The purpose of this section is to go beyond the strict low-rank
assumption by assuming instead that the spectrum admits an
exponentially decaying upper bound.    The main result of this section
is a proof of \thmref{comp.of.q.models.with.spectral.tails}.
\begin{reptheorem}{thm:comp.of.q.models.with.spectral.tails}
\ThmDecayingTails{second}
\end{reptheorem}

\begin{proof}
We will use the same compression scheme as in previous
sections.   As with \thmref{comp.of.q.models},
we use \corref{JL} to argue that the $\eps$-JL property (i.e. $\cJ$) holds with probability $\geq 1 - 4J^2D^2 e^{-\frac{\eps^2}{8}d}$.  

We will also recall the same definition of $\cE_y$ from \eq{event.E.y}, and attempt to prove that it holds with high probability for each $y$.  To this end, fix a particular choice of $y$ and set $E := I - E_{yZ}$.  Let the spectral decomposition of $E$ be
\[
	E = \sum_{j=1}^{\rank(E)} \varepsilon_j \ketbra{\varepsilon_j},
\]
with $\eps_1 \geq \eps_2 \geq \ldots$.
By our assumption in \eq{tail.bound.on.m.j-second}),  after the first $j^*$ eigenvalues, the remaining eigenvalues decay exponentially.  Specifically $\eps_{j^*+j} \leq e^{-bj}$.

Previously we used the low rank of the original POVM elements to
control the fluctuations in the spectrum of the compressed POVM
elements.  Now our POVM elements will in general be full rank.
Moreover, we cannot simply divide them into a low-rank piece and a
low-norm piece, since in the worst case $\|\Pi E \Pi\|$ can be much
large than $\|E\|$.  Instead we split the exponentially decaying part of $E$ into components each of rank
\be r := \frac{1}{b}\ln\frac{8}{\eps} .\label{eq:r-choice}\ee
  (We will see later the reason for this choice.)

\ba 
E & = \sum_{i\geq 0} E^{(i)}  \\
E^{(0)} & = \sum_{j \leq j^* + r}\eps_j \proj{\eps_j} \\
E^{(i)} &= \sum_{j \in j^* + ri + [r]} \eps_j \proj{\eps_j} \text{ for }i > 0,\ea
where $j^* + ri+[r]$ denotes the set $\{j^* + ri+1,\ldots,j^* + ri+r\}$.

We will control the norm $\|\Pi E \Pi^\dag\|$ by bounding each block
separately
\be \| \Pi E \Pi^\dag\| 
 \leq \sum_{i\geq 0} \|\Pi E^{(i)} \Pi^\dag \| 
 \label{eq:Pi-E-triangle}\ee
Now combining \eq{d-beats-tails-second} and \eq{r-choice} we have
 \be d \geq \frac{128}{\eps^2}(j^* + r) \label{eq:d-beats-r}.\ee
By \lemref{control-norm},
\be \prob{\| \Pi E^{(0)}\Pi \| \geq 1+ \eps/2} \leq d e^{-d\eps^2/128}.\ee

We will complete the proof by arguing that the remaining terms in \eq{Pi-E-triangle} have norm $\leq \eps/2$ with high probability.  Using the operator inequality $E^{(i)} \leq
e^{-bri} \sum_{j \in j^*+ri + [r]} \proj{\eps_j}$, we have
\beq\bal\label{eq:sum-gaussian-matrices}
\left\| \sum_{i\geq 1} E^{(i)} \right\|
& \leq \sum_{i\geq 1} e^{-bri} 
\L\|\Pi \sum_{j\in j^*+ri+[r]} \proj{\eps_j} \Pi^\dag \R\|.
\\ & = : 
\sum_{i\geq 1} e^{-bri}  \L\|G_i G_i^\dag \R\| =: X
\eal\eeq
In the second line, we have defined complex $d\times r$-dimensional
matrices $G_i$ which are i.i.d.~and 
each comprised of i.i.d.~complex Gaussians with variance $1/d$,
i.e. as in \thmref{bound.on.sigma.max}.
Define 
\be
\bar X = \frac{e^{-br}}{1-e^{-br}}
(1 + \sqrt{r/d})^2 \ee
Now we bound the moment-generating function of $X$ (assuming $t\leq d/2$) by
\bas
\E[e^{tX}] & =
\prod_{i\geq 1} \E[\exp(te^{-bri}\|G_i G_i^\dag\|] \\
& \leq
\prod_{i\geq 1} \E[\exp(t\|G_i G_i^\dag\|]^{e^{-bri}} \\
& \leq
\prod_{i\geq 1} \L[ d \exp(t(1+\sqrt{r/d})^2 +
\frac{t^2}{d}(1+r/d))\R]^{e^{-bri}}\\
& =
\L[ d \exp(t(1+\sqrt{r/d})^2 +
\frac{t^2}{d}(1+r/d))\R]^{\frac{e^{-br}}{1-e^{-br}}}
\\&=
e^{t\bar X}
 d^{\frac{e^{-br}}{1-e^{-br}}} \exp\L(\frac{t^2}{d}(1+r/d)\frac{e^{-br}}{1-e^{-br}}\R)
\eas
We have used here first the independence of each $G_i$, then the bound $\E[x^\alpha]\leq \E[x]^\alpha$ for $x\geq 0, 0\leq\alpha\leq 1$, then  \thmref{bound.on.sigma.max} and the definition of $\bar X$.
From \eq{r-choice} we have $e^{-br} = \eps/8$ and thus $\frac{e^{-br}}{1-e^{-br}} \leq \eps/4$.  Using also \eq{d-beats-r} we have $(1+r/d)\frac{e^{-br}}{1-e^{-br}} \leq \eps/4$ and $\bar X \leq \eps/4$. 
Now we choose $t=d/2$ and use Markov's inequality to bound
\bas 
\prob{X \geq \bar X + \frac{\eps}{4}} & \leq 
\bbE\L[\exp (tX)\R] e^{-t(\bar X+\frac{\eps}{4})}  \\
&\leq d^{\eps/4} \exp \L( t\bar X + \frac{t^2}{d}\frac{\eps}{4} \R)e^{-t(\bar X+\frac{\eps}{4})} 
\\& \leq d^{\eps/4} \exp[-d\eps/16] 
\\ & = \exp[-d\eps/20].
\eas

Since $\bar X \leq \eps/4$ we conclude that
\be 
\L \| \sum_{i\geq 1} E^{(i)} \R \| \leq \eps/2, \label{eq:tail-small}\ee
with probability $\geq 1-\exp[-d\eps/20] \geq 1-d \exp[-d\eps^2/128]$.  By the union bound, both \eq{tail-small} and $\|\Pi E^{(0)} \Pi^\dag\|\leq 1+\eps/2$ hold with probability
\be \geq 1-2d \exp[-d\eps^2/128] 
\geq 1 - \frac{2d}{4JD}.\ee
This lower bounds the probability of a single $\cE_y$ holding.

By the union bound, we have
\be \prob{\cJ \cup \bigcup_{y\in Y} \cE_y} \geq 
1 - 4J^2D^2 e^{-d\frac{\eps^2}{8}} - \frac{2dY}{4JD} > 0.\ee
The rest of the proof is the same as in
\thmref{comp.of.q.models}.
\end{proof}

\section{Conclusion}

Nonnegative matrices $\data \in \mathbb{R}^{X \times YZ}$ can be described in terms of positive semidefinite factorizations or, in some cases, in terms of quantum models. The number of degrees of freedom in both of these models is a function of their dimension $D$. In practice it is often not necessary to describe $\data$ in terms of an exact psd or quantum factorization but it is sufficient to describe some matrix $\data'$ with $\| \data - \data' \|_{\infty} \leq \delta$. It is then natural to ask for the minimal dimension of such approximate factorizations of $\data$. 

Assume $A_i$ and $B_j$ constitute a $D$-dimensional psd factorization of $\data$. In this paper we showed that in randomized polynomial time we can find $d$-dimensional psd matrices $A_i'$ and $B_j'$ such that 
\beq\bal\label{fwjJIJj324khkfhe}
	\bigl| \data_{i,j} - \tr(A_i'B_j') \bigr| 
	&\leq		192 \varepsilon \;  \tr(A_i) \tr(B_j)\\
	d
	&=		\Bigl\lceil \frac{32}{\varepsilon^2} \ln\bigl( 2JD \bigr) \Bigr\rceil
\eal\eeq
where $J = X + YZ$ (cf. Theorem~\ref{thm:comp.of.psd.factorizations}). Hence, psd factorizations with bounded trace norm admit exponential compression. 

Quantum models are psd factorizations with additional normalization constraints. These normalization constraints can affect compressibility dramatically. Consider for instance $\data = I \in \mathbb{R}^{D \times D}$ with the property $\data_{ij} = \tr\bigl( \ketbra{i} \ketbra{j} \bigr)$. This factorization is both psd and quantum ($Y=1$; $E_{1z} = \ketbra{z}$). As a psd factorization, it admits exponential compression. However, regarded as a quantum factorization, every $\delta$-approximate, $d$-dimensional quantum factorization of $\data$ satisfies 
\[
	d \geq D - Z \delta = (1-\delta) D.
\]
This is an application of the lower bound derived in section~\ref{Sect:incompressibility} (cf. Theorem~\ref{Thm:incompressibility}). Aiming for the compression of quantum models, this observation might seem discouraging. 

However, this lower bound from Theorem~\ref{Thm:incompressibility} only prevents compression below $Z$. It thus remains to investigate compressibility of quantum models with $Z = \mathcal{O}(1)$ in $D$. Introducing a compression scheme which runs in randomized polynomial time in $X,Y$ and $Z$, we showed (cf. Theorem~\ref{thm:comp.of.q.models}) that there exist $d$-dimensional states $\rho'_x$ and POVMs $E'_y$ with approximation promises analogous to~\eqref{fwjJIJj324khkfhe} and 
\beq\label{feowMKMkmwe432}
	d = \max\bigl\{ \frac{16}{\varepsilon^2}  \ln(4JD), \frac{32}{\varepsilon^2} \, \rank\Bigl( \sum_{z=1}^{Z-1} E_{yz} \Bigr) \bigr\}.
\eeq
Note that this compression scheme thus almost achieves the lower bound if $\rank\bigl( \sum_{z=1}^{Z-1} E_{yz} \bigr) \sim Z$. 

Equation~\eqref{feowMKMkmwe432} is only meaningful if the relevant
measurements are low rank. Experimental measurements are however
expected to only be approximately low rank. To cover these scenarios
we derived a version of Theorem~\ref{thm:comp.of.q.models} which can
handle measurements that are only approximately low rank in the sense that their spectrum decays exponentially (cf. Theorem~\ref{thm:comp.of.q.models.with.spectral.tails}).  

In this paper we only started the exploration of applications of these results. We briefly commented on their implications in the fields quantum tomography, dimension witnessing, one-way quantum communication complexity and quantum message identification. In particular our results imply that 
\begin{itemize}
\item high-dimensional models whose measurements are low rank and
\item low-dimensional models
\end{itemize}
are equivalent from an operational perspective if the number of outcomes per measurement is small.

On a less technical and more philosophical side, our results open up a
path towards demystification of the apparent miracle that our
inherently complex, high-dimensional world sometimes admits simple,
low-dimensional descriptions which form the basis of science. Since
our compression schemes are essentially just random projections, these
low-dimensional descriptions can potentially be regarded as random
projections. If that is the case then it should not be surprising that
simple, effective, low-dimensional models can be found if the
measurements data stems from few-outcome and sharp/clean measurements,
i.e., measurements which are approximately low rank.

\section*{Acknowledgment}

We would like to thank Guillaume Aubrun, Matt Coudron, David Gross, Kristan Temme and Henry Yuen for interesting discussions and for being such amazing colleagues. CS acknowledges support through the SNSF postdoctoral fellowship, support from the SNSF through the National Centre of Competence in Research ``Quantum Science and Technology" and funding by the ARO grant Contract Number W911NF-12-0486.  AWH was funded by NSF grant CCF-1111382 and ARO contract W911NF-12-1-0486.


\end{document}